\documentclass[10pt, twocolumn]{IEEEtran}
\usepackage{cite}

\ifCLASSINFOpdf
   \usepackage[pdftex]{graphicx}
  \DeclareGraphicsExtensions{.pdf,.jpeg,.png}
\else
  \usepackage[dvips]{graphicx}
  \DeclareGraphicsExtensions{.eps}
\fi

 \usepackage{nopageno}
\usepackage{relsize}
\usepackage[cmex10]{amsmath}

\usepackage{algorithmic}

\usepackage{array}

\usepackage{dblfloatfix}

\ifCLASSOPTIONcaptionsoff
  \usepackage[nomarkers]{endfloat}
 \let\MYoriglatexcaption\caption
 \renewcommand{\caption}[2][\relax]{\MYoriglatexcaption[#2]{#2}}
\fi

\usepackage{url}

\usepackage{epsfig}
\usepackage{multicol}
\usepackage{verbatim}
\usepackage{amssymb}
\usepackage{epstopdf}
\usepackage{amsthm}
\usepackage{soul}
\usepackage{amsfonts}%
\usepackage{enumerate}%
\usepackage{psfrag}
\usepackage{setspace}
\usepackage{dsfont}
\usepackage{balance}

\usepackage{algorithm}
\usepackage{float}

\usepackage{mdwmath}
\usepackage{mdwtab}
\usepackage[section]{placeins}
\usepackage{caption}
\usepackage{subcaption}
\usepackage{xcolor}
\usepackage{mathrsfs}

\newtheorem{theorem}{Theorem}
\newtheorem{proposition}{Proposition}
\newtheorem{remark}{Remark}

\renewcommand{\vec}[1]{\mathbf{#1}}

\usepackage[font=footnotesize, labelfont={bf,bf}, margin=1cm]{caption}

\begin{document}

\title{On the Secure Degrees of Freedom of the K-user MAC and 2-user Interference Channels }


\author{\IEEEauthorblockN{Mohamed Amir$^1$, Tamer Khattab$^1$ and Tarek Elfouly$^2$}\\
\IEEEauthorblockA{$^1$Electrical Engineering, Qatar University\\
$^2$Computer Science \& Computer Engineering, Qatar University\\
Email: mohamed.amir@qu.edu.qa, tkhattab@ieee.org, tarekfouly@qu.edu.qa}
\thanks{This research was made possible by NPRP 7-923-2-344 grant from the Qatar National Research Fund (a member of The Qatar Foundation). The statements made herein are solely the responsibility of the authors.}}
\vspace{-3mm}
\maketitle
\thispagestyle{empty}
\pagestyle{empty}
\vspace{-3mm}
\begin{abstract}
We investigate the secure degrees of freedom (SDoF) of the $K$-user MIMO multiple access (MAC) and the two user MIMO interference channel. An unknown number of eavesdroppers are trying to decode the messages sent by the transmitters. Each eavesdropper is equipped with a number of antennas less than or equal to a known value $N_E$. The legitimate transmitters and receivers are assumed to have global channel knowledge.  We present the sum SDoF of the two user MIMO interference channel. We derive an upperbound on the sum SDoF of the $K$-user MAC channel and present an achievable scheme that partially meets the derived upperbound.
\vspace{-3mm}
\end{abstract}
\vspace{-3mm}
\section{Introduction}
The noisy wiretap channel was first studied by Wyner \cite{wyner}, in which a
legitimate transmitter (Alice) wishes to send a message to a legitimate receiver (Bob), and hide it from an eavesdropper (Eve). Wyner proved that Alice can send positive secure rate to Bob using channel coding. He derived capacity-equivocation region for the degraded
wiretap channel.  A significant amount of work was carried thereafter to study the information theoretic physical layer security for different network models. The relay assisted wiretap channel was studied in \cite{secop}. The secure degrees of freedom (SDoF) region of multiple access (MAC) channel was presented in \cite{sennur_mac}. The SDoF is the the pre-log of the secrecy capacity region in the high-SNR regime.  Using MIMO systems for securing the message was an intuitive extension due to the spatial gain provided by multiple antennas.  MIMO wiretap channel secrecy capacity was identified in \cite{bab}. Meanwhile, the idea of cooperative jamming was proposed in \cite{yener_coop}, where some of the users transmit independent and identically distributed (i.i.d.) Gaussian noise towards the eavesdropper to improve the sum secrecy rate of the legitimate parties.

In this paper, we study the $K$-user MIMO MAC and the two user MIMO interference channel, each with
unknown number of eavesdroppers. We assume that the legitimate pair has global channel knowledge. We present the sum SDoF of the two user MIMO interference channel. We derive an upperbound on the the sum SDoF of the $K$-user MAC channel and present an achievable scheme that partially meets the upperbound depending on the relations between the nodes' number of antennas. We use the following notation, $\vec{a}$ for vectors, $\vec{A}$ for matrices, $\vec{A}^{\dagger}$ for the hermitian transpose of $\vec{A}$, $[A]^+$ for the $\max{A,0}$ and $\mathrm{Null}(\vec{A})$ to define the nullspace of $\vec{A}$, while$a\textbf{C}b$ is used to define the $b$-combination of a set $a$


\section{System model}
\label{sec:model}

We consider two communication systems, the $K$-user MIMO MAC and the two user MIMO interference channel. The $K$-user MIMO MAC consists of $K$ transmitters, each is equipped with $M$ antennas and one legitimate receiver equipped with $N$ antennas. The two user MIMO interference channel consists of two transmitters and two receivers, each is equipped with $M$ antennas.  Both systems are studied in vicinity of an unknown number of passive eavesdroppers.  The $j$th eavesdropper is equipped with $N_{Ej}\leq N_E$ antennas, where $N_E$ is a constant known to all transmitters. Let $\vec{x}_i$ denote the $M \times 1$ vector of symbols to be transmitted by transmitter $i$, where $i\in\{1,2,...,K\}$. We can
write the received signal at the $j$th legitimate receiver at time (sample) $k$ as

\begin{equation}\label{Received_signal}
\vec{Y}_j(k)=\sum_{i=1}^q\vec{H}_{i,j} \vec{V}_i\vec{x}_i(k)+ \vec{n}_j(k), 
\end{equation}

where $i \in \{1,2,...,K\}, j=1 \text{ and }  q=K$ for the MAC channel, $i, j \in \{1,2\}\text{ and } q=2$ for the interference channel and the received signal at the $j$th eavesdropper is

\begin{equation}\label{Received_signal}
\vec{Z}_{j}(k)=\sum_{i=1}^q\vec{G}_{i,j}(k) \vec{V}_i\vec{x}_i(k)+ \vec{n}_{Ej}(k),
\end{equation}
\noindent where $\vec{H}_{i,j}$ is the matrix containing
the channel coefficients from transmitter $i$ to the legitimate receiver $j$,  $\vec{G}_{i,j}(k)$ is the matrix containing the channel coefficients from transmitter $i$ to the eavesdropper $j$,  $\vec{V}_i$ is the precoding unitary matrix (i.e. $\vec{V}_i\vec{V}_i^\dagger = \vec{I}$) at transmitter $i$, $\vec{n}_j(k)$ and $\vec{n}_{Ej}(k)$ are the additive white Gaussian
noise vectors with zero mean and variance $\sigma^2$ at the legitimate receiver and the $j$th eavesdropper, respectively. We assume that the transmitters have global channel knowledge. We assume that $N_E< M$.  We define the $M \times 1$ channel input from legitimate transmitter $i$ as
\begin{equation}
\vec{X}_i(k)= \vec{V}_i \vec{x}_i(k).
\end{equation}
\noindent Each transmitter $i$ intends to send a message $W_i$ over $n$ channel uses (samples) to the legitimate receiver simultaneously while preventing the eavesdroppers from decoding its message.  The encoding occurs under  a constrained power given by
\begin{equation}
\text{E}\left\{tr(\vec{X}_i\vec{X}_i^{\dagger})\right\} \leq P \text{  }\forall{i=1,...,q}
\end{equation}
\noindent Expanding the notations over $n$ channel extensions we get $\vec{H}_i^n\hspace{-1mm}=\hspace{-1mm}\{\vec{H}_i(1), \vec{H}_i(2),\ldots, \vec{H}_i(n)\}$. Similarly we can define $\vec{G}_{i,j}^{n}, \vec{X}_i^n,\vec{Y}^n, \vec{Z}_{j}^n$. At each transmitter, the message $W_i$ is uniformly and independently chosen from a set of possible secret messages for transmitter $i$, $\mathcal{W}_i = \{1,2, \ldots, 2^{nR_i}\}$.  The rate for $W_i$ is $R_i \triangleq \frac{1}{n} \log\left|\mathcal{W}_i\right|$, where $|\cdot|$ denotes the cardinality of the set.
A secure rate tuple $(R_1,...,R_q)$ is said to be achievable if for any $\epsilon > 0$ there is an $n$-length codes such that the legitimate receiver decode the messages reliably, i.e.,
\begin{equation}
\text{Pr}\{(W_1,..., \hat{W}_q)) \neq (\hat{W}_1,..., \hat{W}_q)\} \leq \epsilon
\end{equation}
and the messages are kept information-theoretically secure against the eavesdroppers, i.e.,
\begin{equation}\label{eqn:cond}
\lim\limits_{n\longrightarrow\infty} \frac{1}{n}H(W_1,..., W_q|\vec{Z}_{j}^{n})\geq \lim\limits_{n\longrightarrow\infty} \frac{1}{n} H(W_1,..., W_q)-\epsilon \\
\end{equation}
\noindent where $H(\cdot)$ is the Entropy function and~\eqref{eqn:cond} implies the secrecy for any subset $\mathbb{S} \subset \{1,2\}$ of messages including individual messages \cite{sennur_mimo}. The sum SDoF is defined as
\begin{equation}
D_s = \lim_{P\rightarrow \infty} \sup{\sum_i \frac{R_i}{\frac{1}{2}\log P}},
\end{equation}
\noindent where the supremum is over all achievable secrecy rate tuples $(R_1,..., R_q)$, $D_s = d_1 +...+ d_q$, and $d_{i}$ is the secure DoF of transmitter $i$. 


\section{K User MIMO MAC}\label{sec:bound}
\begin{theorem}
The number of SDoF of the K user MAC channel is upperbound as,
\small
\begin{eqnarray}
\label{up}D_s\leq \begin{cases} \min (KM-N_E, \; N-\frac{N_E}{K} ) & \text{   if } M < N\\
M-\frac{N_E}{K} & \text{   if } N \leq M < N+\frac{N_E}{K}\\
N & \text{    if } M\geq N+\frac{N_E}{K}\\
\end{cases} 
\end{eqnarray}
\normalsize
\end{theorem}
\begin{proof}

The first bound for $D_s \leq KM-N_E$ represent the DoF loss caused by the number of eavesdroppers' antennas on the transmitter side. Without loss of generality, we provide an upperbound for the case of existence of only one eavesdropper with $N_E$ antennas. The SDoF of the single eavesdropper scenario is certainly an upperbound for the multiple eavesdroppers case, as increasing the number of eavesdroppers can only reduce the SDoF of the legitimate users. Accordingly, we omit the eavesdropper subscript for simplicity of notation. Suppose that we can added $|M-N|^+$ antennas to the receiver side that won't decrease the SDoF, the sum rate is upperbounded by the capacity of an equivalent MIMO wiretap channel with $(M_1+M_2)$ transmit antennas and $\vec{H}=[\vec{H}_1 \text{ } \vec{H}_2 ]$, $\vec{x}=[\vec{x}_1 \text{ }\vec{x}_2 ... { }\vec{x}_K]^T$ and precoding matrix $\vec{V}$. The secrecy capacity ($C_s$)for the MIMO wiretap channel with one eavesdropper and fixed known eavesdropper channel  was presented in \cite{bab}, and is an upperbound for all cases studied in this paper. It is easy to see that if the eavesdropper channel is unknown and time varying the SDoF is also upperbounded by the fixed channel case. The secrecy capacity ($C_s$) was found to be equal to,
\small
\begin{eqnarray}
C_s&=&(\vec{X}_1, \vec{Y}_1 ) - I(\vec{X}_1, \vec{Z} )\\
&=&\max_{K_x} \log (|(\vec{I}+\vec{H}_{1,1} K_x \vec{H}_{1,1}^{\dagger})| -\log |(\vec{I}+\vec{G} K_x \vec{G}^{\dagger})|)
\end{eqnarray}
\normalsize
where $K_x$ is the covariance matrix of the transmitted signal. As $\vec{H}^{\dagger}_{11}\vec{H}_{11}$ and $\vec{G}^{\dagger}\vec{G}$ are hermitian, they can be diagonalized as $\vec{G}^{\dagger}\vec{G}= \vec{U}_G\vec{\Lambda}_G\vec{U}_G^{\dagger}$, $\vec{H}^{\dagger}_{11}\vec{H}_{11}= \vec{U}_{H_{11}}\vec{\Lambda}_{H_{11}}\vec{U}_{H_{11}}^{\dagger}$, where $\vec{U}_G\vec{U}_G^{\dagger}=\vec{I}$ and $\vec{U}_G\vec{U}_G^{\dagger}=\vec{I}$. Without loss of generality, Let $\vec{V}= [\vec{V}_L \vec{V}_N]$, where $\vec{V}_N$ contains the $N_E$ orthonormal basis of $\vec{G}$, while $\vec{V}_L$ contains the $M-N_E$ basis of the orthogonal complement of $\vec{V}_N$, and  $K_x = \vec{V}\vec{\Lambda}_{K_x}\vec{V}^{\dagger}$.  Therefore,
\vspace{-.5mm}
\small
\begin{eqnarray}
\label{eq2}\nonumber D_s\hspace{-3mm}&\leq& \hspace{-2.5mm}\lim_{P\rightarrow \infty} \frac{1}{\log P} \big(\max_{\vec{\Lambda}_{K_x})} (\log | \vec{I}+\vec{U}_{H_{11}}\vec{\Lambda}_{H_{11}}\vec{U}_{H_{11}}^{\dagger} \vec{V}\vec{\Lambda}_{K_x}\vec{V}^{\dagger} | \\
\nonumber &-&\hspace{-2.5mm}\log | \vec{I}+\vec{U}_G\vec{\Lambda}_G \vec{U}_G^{\dagger}\vec{V}\vec{\Lambda}_{K_x}\vec{V}^{\dagger}|)\big)\\
\nonumber&\overset{(a)}{\leq}&\hspace{-2.5mm}\lim_{P\rightarrow \infty}\frac{1}{\log P} \big(\max_{\vec{\Lambda}_{K_x}}( \log  |\vec{\Lambda}_{H_{11}}\vec{\Lambda}_{K_x}| -\log  | \vec{\Lambda}_G \vec{\Lambda}_{K_x}| \big)\\
\nonumber &\overset{(b)}{\leq}&\hspace{-2.5mm}\lim_{P\rightarrow \infty} \frac{1}{\log P}\big(\max_{\vec{\Lambda}_{K_x}} (\log\prod_{i=1}^{KM}{\lambda_{{H}_{11}}^i\lambda_{K_x}^i}- \log \prod_{i=1}^{N_E}{\lambda_{G}^i\lambda_{K_x}^i})\big)\\
\nonumber &\leq&\hspace{-2.5mm}\lim_{P\rightarrow \infty} \frac{1}{\log P}\big(\max_{\vec{\Lambda}_{K_x}} (\sum_{i=1}^{KM}\log{\lambda_{{H}_{11}}^i\lambda_{K_x}^i}- \sum_{i=1}^{N_E}\log {\lambda_{G}^i\lambda_{K_x}^i})\big)\\
\label{up2}&\leq& \hspace{-2.5mm} KM-N_E 
\end{eqnarray}
\normalsize
where $\lambda^i_{K_x}$ is the $i$th diagonal value of $\Lambda_{K_x}$ and $\lambda_{G}^i, \lambda_{H}^i$ are defined similarly. (a) is because $\log |\vec{I}+\vec{A}\vec{B}|= \log |\vec{I}+\vec{B}\vec{A}|$ for the above matrices, (b) is because $\lim\limits_{P\rightarrow \infty} \frac{\log |\vec{I}+\vec{B}|}{\log P} = \lim\limits_{P\rightarrow \infty} \frac{\log |\vec{B}|}{\log P}  $ for any matrix $\vec{B}$, and because $|\vec{A}\vec{B}|= |\vec{A}||\vec{B}|$ for square matrices, and ${|\vec{V}_{K_x}|, |\vec{V}_{H_{11}}|, |\vec{U}| }$ are independent of $P$.
\end{proof}
\vspace{-3mm}
\noindent The second bound $M-\frac{N_E}{K}$ represents the DoF loss of each transmitter due to the number of eavesdroppers antennas available. Let $d_e^i$ be degrees of freedom of the parts of the messages sent by transmitter $i$, which can be decoded by the eavesdropper. For the receiver to be able to decode the secure messages with inter-message interference and achieve the designated SDoF of each transmitter, the receiver must be able to project the $i$th secure signal into an interference free space of dimension $d_i$. On the other hand, the non secure parts of the messages can overlap at the receiver or even does not reach the receiver because the receiver is not interested in decoding them and treated as interference. Let $\alpha_i$ be the number of degrees of freedom of the non secure part of the message $i$ that reaches the receivers while $\beta_i$ be the number of degrees of freedom of the part that does not reach the receiver. Accordingly, $d_e^i= \alpha_i + \beta_i$ and the number of degrees of freedom of the message $i$ is equal to $(d_i+ \alpha_i + \beta_i)$. Since, the receiver is not interested in decoding the non secure parts with sizes $\{\alpha_i; i= 1, 2,..., K\}$, and the non secure messages occupy at least $\max(\alpha_j; j = 1, 2,...,K)$ DoF, then
\begin{equation}\label{z1}
\sum^K_{i=0}d_{i}+\max(\alpha_j; j = 1, 2,...,K) \leq N
\end{equation}
while,
\begin{equation}\label{z1}
\beta_i\leq M-N \text{ } \forall \text{ } i=1, 2,..., K
\end{equation}

Moreover, since the eavesdropper has $N_E < M$ antennas, i.e the Dof of the transmitted messages is larger than $N_E$ then 
\begin{equation} \label{eve}
\sum^K_{i=0}d_{e}^i=N_E
\end{equation}

The secure DoF is then upperbounded as 

\small
\begin{eqnarray}
D_s\leq&\text{maximize} &{\{d_j; j = 1, 2,...,K\}}\sum^K_{i=0}d_{i}\\
&\text{subject to,} &\sum^K_{i=0}d_{i}+\max(\alpha_j; j = 1, 2,...,K) \leq N\\
\nonumber&&\sum^K_{i=0}(\alpha_i +\beta_i)=N_E\\
\nonumber&&\sum^K_{i=0}\beta_i \leq K(M-N)
\end{eqnarray}
\normalsize
The sum SDoF is maximized by minimizing $\max(\alpha_i; j = 1,2,...,K)$. Combining the second and third constraint we have 
\begin{equation}
\sum^K_{i=0}\alpha_i \geq N_E - K(M-N)\\
\end{equation}

and minimizing $\max(\alpha_i; j = 1,2,...,K)$ is achieved by choosing all $\{\alpha_i; j = 1,2,...,K\}$ to be equal, and $\sum^K_{i=0}\alpha_i= N_E - K(M-N)$. Accordingly, 
\begin{equation}
\max(\alpha_j; j = 1, 2,...,K) \leq  \frac{N_E}{K} - (M-N) \\
\end{equation}
and,
\begin{equation}\label{c3}
D_s\leq M-\frac{N_E}{K}
\end{equation}
Similarly, we can prove that for $M\leq N$, the SDoF is upperbounded as,

\begin{equation}\label{c3}
D_s\leq N-\frac{N_E}{K}
\end{equation} 
\normalsize

The Third bound $D_s \leq N$ is the due to limited number of antennas at the receiver which limits the SDoF.\\
\normalsize

\noindent\textit{Achievable scheme}:\\
For securing the legitimate messages, the transmitters uses a two-step noise injection by simultaneously sending a jamming signal and using a stochastic encoder as follows,

\begin{enumerate}
\item The transmitters send a jamming signal with power $P^J=\alpha P$ that guarantees that all eavesdropper have a constant rate ($o(log P)$) for all legitimate signal power values,  where $\alpha$ is a constant controlled the transmitters to adjust the jamming.
\item A stochastic encoder is built using random binning. The encoder randomness rate is designed to be larger than of the post-jamming eavesdroppers leakage, hence all eavesdroppers would have zero rate with the code length goes to infinity meeting the secrecy constraints in \eqref{eqn:cond}.
\end{enumerate} 
The jamming signal transmitted is a $N_E$ vector $\vec{r}=[\vec{r}_1 \text{ } \vec{r}_2 ... \text{ } \vec{r}_K]^T$ with random symbols using \{$\vec{V}^J_{1}$, $\vec{V}^J_{2}$, , $\vec{V}^J_{K}$\} as jamming precoders\footnote{For the special case $N_E=1$, only one user sends a single jamming symbol.}. Hence, the transmitted coded signal can be broken into legitimate signal, $\vec{s}_i$, and jamming signal, $\vec{r}_i$, the precoder, $\vec{V}_i$ can be also broken into legitimate precoder, $\vec{V}^L_i$, and jamming precoder, $\vec{V}^J_i$ such that
$$ \vec{V}_i\vec{x}_i = \left[\begin{array}{cc} \vec{V}^L_i &  \vec{V}^J_i \end{array}\right]  \left[\begin{array}{c} \vec{s}_i\\  \vec{r}_i \end{array}\right],  \in\{1,2,...,K\}.$$
Choosing $\vec{V}^J$ to be the unitary matrix, the jamming power becomes $P^J=\text{E}\{\text{tr}(\vec{r}_i\vec{r}_i^{\dagger})\} = \alpha P$, where $\alpha$ is a constant controlled by the transmitter.
\vspace{-1.5mm}
\begin{proposition}
The jamming signal, $\vec{r}$, overwhelms \textit{all} eavesdroppers' signal space, and all eavesdroppers end up decoding zero DoF of the legitimate messages. The transmitter then uses a stochastic encoder to satisfy the secrecy constraint in \eqref{eqn:cond}
\end{proposition}

Let $\bar{R}_e= I(\vec{Z}; \vec{s}_1, \vec{s}_2,..., \vec{s}_K)$ be the rate of the eavesdropper with the best channel assuming in worst case scenario that it also has $N_E$ antennas. Let $R_e = I(\vec{Z}; \vec{W}_1, \vec{W}_2,...,\vec{W}_K)$ be the legitimate message rate of the same eavesdropper, where $R_e < \bar{R}_e$ because of the stochastic encoder used. let $\bar{R}_ej$ be the rate of the $j$th eavesdropper. Then $\bar{R}_ej\leq \bar{R}_e \forall j \in L$, where $L$ is the unknown number of eavesdroppers.
\begin{proof}
\small
\begin{eqnarray}
\nonumber n \bar{R}_e &\leq& I(\vec{Z}^n; \vec{s}_1^n, \vec{s}_2^n,...,\vec{s}_K^n)\\
\nonumber  & =& h(\vec{Z}^n) -h(\vec{Z}^n|\vec{s}_1^n, \vec{s}_2^n,...,\vec{s}_K^n)\\
\nonumber \bar{R}_e  & =& h(\vec{Z}) -h(\vec{Z}|\vec{s}_1, \vec{s}_2,...,\vec{s}_K)\\
\nonumber &\leq& N_E \log P -N_E log P^J + o(\log P)  \\
\nonumber     &\leq& N_E \log P -N_E \log \alpha P +o(\log P) \\
\label{eve}&\leq& o(\log P)=C_E  
\end{eqnarray}
\normalsize
\noindent where $C_E$ is a constant that does not depend on $P$ and known to the transmitters. 
\end{proof}
\begin{remark}
The constant eavesdropper post-jamming rate comes from the fact that $P^J$ is controlled by the transmitter.  Hence, setting $P^J=\alpha P$, a constant SNR is guaranteed at the eavesdroppers and a constant rate independent of $P$. For the case of the constant known eavesdropper channel or unknown fading channel with known statistics, the constant $C$ is known transmitter.
\end{remark}

The transmitters use the post-jamming rate difference to transmit perfectly secure messages using a stochastic encoder similar to the one described in~\cite{khan} according to the strongest eavesdropper's rate, $C$, in worst case scenario to achieve the secrecy constraint in \ref{eqn:cond} .
The Wyner code $C_i \in C(R^t_i, R_i, n) \forall i=1,2,...,K$ of size $2^{nR^t_i}$ is used to encode a confidential message set $W_i = \{1, 2, . . . , 2^{nR_i} \}$ of transmitter $i$, $R^t_i$ is the transmitted total rate and $R_i$ the secure message rate (i.e. $R^t_i \geq R_i$), and $n$ is the codeword length. As a result, the rate $R^l=R^t_i - R_i$ is the cost of secrecy or the rate lost to secure the legitimate message. For a Wyner code, if $\hat{R}_e= R_l$, then the eavesdropper cannot decode the secure message sent (i.e $\lim_{n\longrightarrow \infty} \frac{1}{n} R_e \leq \epsilon)$.
The Wyner code $C(R^t_i, R_i, N)$ is built using random binning \cite{9}. We generate $2^{nR^t_i}$ codewords $s_i^n(w_i, v_i)$, where $w_i = 1, 2, . . ., 2^{nR_i} $,
and $v_i = 1, 2, . . ., 2^{n(R^t_i-R_i)}$, by choosing the $2^{nR_i^t}$ symbols $s_i(w_i, v_i)$ independently at random according to the input distribution $p(s_i)$. Then we distribute them randomly into  $2^{nR_i}$ bins such that each bin contains $2^{n(R_i^t-R_i)} $ codewords.
The stochastic encoder of $C(R_i^t, R_i, N)$ is described by a matrix of conditional
probabilities so that, given $w_i \in W_i$, we randomly and uniformly select a codeword to transmit  from the bin $w_i$ or in other words, we select $v_i$ from
$\{1, 2, . . . , 2^{n(R^t_i-R_i)}\}$ and transmit $s_i^n(w_i, v_i)$. We assume that the legitimate
receiver employs a typical-set decoder. Given the received signal $y^n$, the legitimate
receiver tries to find a pair $(\hat{w} , \hat{v})$ so that $s^n(\hat{w} , \hat{v})$ and $y^n$ are jointly typical \cite{9}. We set $R_i= I(\vec{s}_i, \vec{Y}_1)- I(\vec{s}_i,\vec{Z}) -\epsilon$ and  $R_i^t=I(\vec{s}_i, \vec{Y}_1)-\epsilon$. The error probability and equivocation calculations are
straight forward extensions of similar Wyner random binning
encoders \cite{9},

\begin{eqnarray}
H(\vec{W_i}^n)\hspace{-2mm} &=& I(\vec{s}_i^n; \vec{Y}^n) - I(\vec{s}_i^n;\vec{Z}^n) -m\epsilon\\
H(\vec{W_i}^n|\vec{Z}^n)&=&I(\vec{s}_i^n; \vec{Y}^n|\vec{Z}^n) - I(\vec{s}_i^n;\vec{Z}^n|\vec{Z}^n)-n\epsilon\hspace{4mm}\\
&=& I(\vec{s}_i^n; \vec{Y}^n, \vec{Z}^n) - I(\vec{s}_i^n,\vec{Z}^n)-n\epsilon\\
&\geq& H(\vec{W_i}^n) -n\epsilon
\end{eqnarray}
and,
\small
\begin{eqnarray}
H(\vec{W_1}^n,..., \vec{W_K}^n|\vec{Z}^n)\hspace{-2mm} &=& \hspace{-2mm}\sum_i^K H(\vec{W_i}^n|\vec{Z}^n)\hspace{2mm}\\
&\geq& \hspace{-2mm} \sum_i^K H(\vec{W_i}^n) -Kn\epsilon\\
&\geq& \hspace{-2mm} H(\vec{W_1}^n, ...,H(\vec{W_K}^n) -Kn\epsilon \hspace{2mm}
\end{eqnarray}. 
\normalsize
Let $\vec{U}$ be the post-processing matrix that projects the signal into a jamming free space at the legitimate receiver. The secure messages sum rate is then,

\footnotesize
 \begin{eqnarray} \label{secrate}
\sum_{i=1}^K R_i \hspace{-7mm}&& \geq \frac{1}{2} \log \left| \vec{I}+ \sum_{i=1}^K (\vec{U}\vec{H}_i\vec{V}_i^L\vec{s}_i\vec{s}^{\dagger}_i\vec{V}^{L\dagger}_i\vec{H}_i^{\dagger}\vec{U}^{\dagger})\right| -R_e
\end{eqnarray}
\normalsize

As $\lim_{n \longrightarrow \infty} \frac{1}{n} R_e \leq \epsilon$ for all values of $\vec{G}_i$ and $P$, a positive secrecy rate, which is monotonically increasing with $P$, is achieved. Computing the secrecy degrees of freedom boils down to calculating the degrees of freedom for the first term in the right hand side of \eqref{secrate}, which represents the receiver DoF after jamming is applied. Next we will calculate the SDoF and show how jamming is designed to Maximize the achievable SDoF.

\subsection{Achievability for $M \geq N+\frac{N_E}{K}$}

For this region, the transmitters send the jamming signals using precoders $\vec{V}^J_{i}$, using \textit{Nullspace jamming}  , respectively. All the precoders have $\frac{N_E}{K}$ jamming streams such that the total number of jamming streams reaching each eavesdropper equal $N_E$. 

\subsubsection*{Nullspace jamming}
In nullspace jamming method, the transmitter $i$ sends a jamming signal of $J_i$ dimensions using the precoder $\vec{V}^J_{i}$ which lies in the nullspace of the channel $\vec{H}_i$,
\vspace{-2mm}
\begin{eqnarray}
\vec{V}^J_i=\text{\textit{Null}} (\vec{H}_i) \text{ } i\in 1,2,....,K
\end{eqnarray}
\noindent This blocks $N_E$ dimensions at each eavesdropper and leaves $N$ free dimensions at the legitimate receiver to attain the legitimate signal, thus the following sum SDoF is achievable,
\begin{eqnarray}
D_s &\leq& N
\end{eqnarray}

\subsection{Achievability for $M<N$}
For this region we use \textit{aligned jamming} for blocking the eavesdropper where jamming is aligned at the receiver to minimize the wasted space and maximize the SDoF

\subsubsection*{Aligned jamming}
The jamming signals of both transmitters are aligned at the legitimate receiver signal space. Each group $j$ of transmitters of size $L_j\leq K$ aligns portions of its jamming signal together at the receiver. There are $K\boldsymbol{C}L_j$ groups of size $L$, while the total number of groups $\sum_{a=2, b\leq a}^{a=K} i\boldsymbol{C}j$, the number of jamming of streams assigned to each group $(J_g/N_E)$ depends on the relation between $(M, N, N_E)$.
Let $\mathcal{I}_j$ be the jamming space at the receiver designated for group $j$. Each transmitter aligns a part or the whole of its jamming signal into this jamming space. The total signal space of transmitter $i$ occupies  \emph{only} $M<N$ dimensions at the receiver. This make the received signal spaces of different transmitters distinct at the receiver. So a common space is needed to direct the jamming signal into. Let $A_i$ span the received signal space of transmitter $i$, i.e span the space including all possible received vectors at the receiver, $\mathcal{I}$ is chosen to be the intersection of these spaces, i.e.,
\begin{equation}\label{first}
\mathcal{I}_j= \bigcap_{i=1}^{L_j} A_i.
\end{equation}
$\mathcal{I}_j$ would have positive size only if $M\geq N$.
Without loss of generality, we design $(\vec{V}_i^J,\text{ } i=1,2,...,K)$ such that,
\begin{eqnarray}\label{two}
\vec{H}_{1,1}\vec{V}^J_1= \vec{H}_{2,1}\vec{V}^J_2=...= \vec{H}_{L_j,1}\vec{V}^J_{L_j}= \mathcal{I}_j
\end{eqnarray}
While the system of equations in (\eqref{two}) has more variables than the number of equations,~\eqref{first} ensures that the system has a unique solution as $\mathcal{I}_{j}$ lies in the spans of $(\vec{H}_{i,1}; i=1,2,...,K$.

Let
\begin{equation}
 \vec{H}_{i,1}=
\begin{bmatrix}
 \vec{H}_{i,1}^{'}\\
 \vec{H}_{i,1}^{''}
\end{bmatrix}  \mathcal{I}=
\begin{bmatrix}
 \mathcal{I}_j' \\
\mathcal{I}_j''
\end{bmatrix}
\;\; \forall \; i=1,2,...,K,\text{ } j=1,2,..,L_j
\end{equation}
where $\vec{H}_{i,1}'$ contains the  $M$ rows of $\vec{H}_{i,1}$ and $\vec{H}_{i,1}''$ contains the other $N-M$ rows, and $\mathcal{I}_i'$ contains the  $M$ rows of $\mathcal{I}$ and $\mathcal{I}_i''$ contains the other $N-M$ rows. Therefore, we can choose the following design which satisfies ~\eqref{two}
\begin{equation}\label{last}
\vec{V}^J_i =(\vec{H}_{i,1}')^{-1}\mathcal{I}_j' \;\forall\text{ } i=1,2,...,K,\text{ } j=1,2,..,L_j
\end{equation}
\noindent For the legitimate receiver to remove the jamming signal and decode the legitimate message, it zero forces the jamming signal using the post-processing matrix $\vec{U}$. For the case $N_E$ is odd, each transmitter will align its jamming signal into an $\lfloor \frac{N_E}{2} \rfloor$--dimensional half space using linear alignment.  The remaining $1$dimensional space will be equally shared between the two transmitters' jamming signal using real interference alignment~\cite{sennur_helpers}, yielding each transmitter's jamming signal to occupy $\frac{N_E}{2}$.

The jamming alignment is possible for group $j$ the size of intersection of $j$ is greater than zero or 
\begin{equation}
L_j(M-N)+M \geq 0
\end{equation}

where the number of jamming streams $J_j$ that can be sent by each group is constrained to

\begin{equation}
J_j \leq L_j( L_j(M-N)+M) 
\end{equation}
where the $J_j$ streams wastes $\frac{J_j}{L_j}$ dimensions at the receiver for jamming. For maximizing the achievable SDoF, the group with the smaller ratio $\frac{J_j}{L_j}$ is used for jamming first.

The alignment process begins with assigning the maximum number of jamming streams to the largest possible group as it can align the largest number of jamming streams per one dimension wasted at the receiver.

\begin{equation}
D_s\leq \min(KM-NE, \; N-\frac{N_E}{L}),  
\end{equation}
where $L$ is the result of the following optimization,
\begin{eqnarray}\label{Ldef}
&\text{maximize} &{L}\\
&\text{subject to,} &\frac{M}{N-M}\leq L \leq K\\
&&N_E\leq L(L(M-N+M))
\end{eqnarray}

The previous scheme meets the upperbound in \ref{up} at (($L=K$ or  $\frac{M}{N-M}\leq K$) and $N_E\leq K(K(M-N)+M)$) and at ($KM-N_E< N-\frac{N}{N_E}$) achieving $D_s=\min(KM-N_E, N-\frac{N}{N_E})$ .

\subsection{Achievability for $ N+\frac{N_E}{K}> M>N$ }

In this region the transmitters uses both the aligned and Nullspace jamming methods, each transmitter sends $M-N$ jamming streams using Nullspace jamming and sends $\frac{N_E}{K}-(M-N)$ jamming streams using aligned jamming.

The achievable SDoF is then

\begin{equation}
D_s \leq N-\frac{N_E-K(M-N)}{L}
\end{equation}

\begin{equation}
D_s \leq N-\frac{N_E-K(M-N)}{L}
\end{equation}

where $L$ is defined as in \eqref{Ldef}, and the achievable region meets the upperbound at ($L=K$ and $N_E\leq K((K+1)(M-N)+M)$) and at $(K(2M-N)-N_E< N-\frac{N}{N_E})$.

\begin{eqnarray}
D_s &\leq& N-\frac{N_E-K(M-N)}{K}\\
&\leq& M-\frac{N_E}{K}
\end{eqnarray}

\section{The two user $M \times M$ interference channel}\label{sec:IC}
\begin{theorem}
The number of SDoF of the two user $M \times M$ interference channel is upperbound as,
\small
\begin{equation}
\! d_1+d_2\!\leq \!  M-\frac{N_E}{2}
\end{equation}
\normalfont
\end{theorem}
\begin{proof}
  Let $d_e^1$ and $d_e^2$ be the maximum degrees of freedom that the eavesdropper can decode out of the transmitters one and two signals, respectively. Suppose that we added $M-N$ antennas to receiver one, this can only improve the coding scheme rate. As receiver one fully receive the signal sent by transmitter two to receiver two $X_2$ with modified noise, and $X_1$ can be decoded by receiver one with no interference by definition. Then $d_1$ and $X_2$occupies two distinct spaces at receiver one, the SDoF is upperbounded then as 
\begin{equation}\label{int1}
d_{1}+\max(d_{e}^1 ,d_{2}+d_{e}^2) \leq M
\end{equation}
and for both results of $\max(d_{e}^1 ,d_{2}+d_{e}^2)$, the following is true 
\begin{equation}\label{z7}
d_{1}+d_{2}+d_{e}^2 \leq M
\end{equation}
\\
Similarly, by adding $M-N$ antennas to receiver two, we have
\begin{equation}\label{z2}
d_2+ d_{1}+d_{e}^1  \leq M 
\end{equation}

Moreover, since the eavesdropper has $N_E$ antennas then 
\begin{equation}\label{z3}
d_{e}^1+d_{e}^2=N_E
\end{equation}

Combining (\ref{z1}), (\ref{z2}), (\ref{z3})
\begin{equation}\label{z4}
d_1+d_2 \leq M-\frac{N_E}{2} 
\end{equation}
\end{proof}

\begin{theorem}
For the two user $M \times M$ interference channel, the following number of SDoF is achievable
\begin{equation}
\! d_1+d_2\!\leq \! M-\frac{N_E}{2}
\end{equation}
\end{theorem}
\begin{proof}

For this channel, the jamming is aligned using basic interference alignment method combined with a stochastic encoder similar to the one used in the MAC,
\begin{eqnarray}
\boldsymbol{H}_{21} \boldsymbol{V}_1^J= \boldsymbol{H}_{22} \boldsymbol{V}_2^J\\
\boldsymbol{H}_{11} \boldsymbol{V}_1^J= \boldsymbol{H}_{12} \boldsymbol{V}_2^J
\end{eqnarray}

Using this method $N_E$ jamming streams are aligned at $\frac{N_E}{2}$ dimensions at each receiver.
This leaves $N-\frac{N_E}{2}$ dimensions free of Jamming at each receiver. As both receivers fully receives both messages and for each receiver to decode its own message the interfering message should occupy an orthogonal space, then 

\begin{equation}
d_1+d_2 \leq M-\frac{N_E}{2}\\
\end{equation}
\vspace{-2mm}
\end{proof}  
\vspace{-5mm}
\section{Conclusion}
\label{sec:conclusion}
We studied the $K$-user MAC channel  and the two user interference channel with multiple antennas at the transmitters, legitimate receivers and eavesdroppers. Generalizing new upperbound was established and  a new achievable scheme was provided. We showed that our scheme is optimal for the interference channel and partially optimal for the MAC.

\end{document}